\documentclass[peer review]{IEEEtran}

\IEEEoverridecommandlockouts

\usepackage[cmex10]{amsmath}
\usepackage{amssymb,amsthm,color,graphicx,algorithm,algorithmic}

\newtheorem{thm}{Theorem}
\newtheorem{lem}[thm]{Lemma}

\newtheorem{cor}[thm]{Corollary}
\newtheorem{defini}[thm]{Definition}
\newenvironment{defi}{\begin{defini}\rm}{\end{defini}}

\theoremstyle{definition}
\newtheorem{ex}[thm]{Example}
\newtheorem{rema}[thm]{Remark}
\newenvironment{rem}{\begin{rema}\rm}{\end{rema}}

\newcommand{\field}[1]{\mathbb{#1}}
\newcommand{\F}{\field{F}}

\newcommand{\cM}{{\mathcal M}}

\newcommand{\PG}{\mathcal{P}_{q}(n)}
\newcommand{\Gr}{\mathcal{G}_{q}(k,n)}
\newcommand{\G}{\mathcal{G}_{q}(k,n)}
\newcommand{\Uvs}{\mathcal{U}}
\newcommand{\Vvs}{\mathcal{V}}
\newcommand{\Cvs}{\mathcal{C}}
\newcommand{\C}{\mathcal{C}}
\newcommand{\Rvs}{\mathcal{R}}
\newcommand{\rs}{\mathrm{rs}}
\newcommand{\rk}{\mathrm{rank}}
\newcommand{\enc}{\mathrm{enc}}
\newcommand{\GL}{\mathrm{GL}}
\newcommand{\stab}{\mathrm{stab}}
\newcommand{\ord}{\mathrm{ord}}

\begin{document}

\title{Message Encoding and Retrieval for Spread and Cyclic Orbit Codes}

\author{\IEEEauthorblockN{Anna-Lena Horlemann-Trautmann}\\
\IEEEauthorblockA{Algorithmics Laboratory\\
EPF Lausanne, Switzerland\\
Email: anna-lena.horlemann@epfl.ch}\\
\thanks{Parts of this work were presented at the IEEE International Symposium on Information Theory (ISIT) 2014 in Honolulu, USA, and appear in its proceedings \cite{tr14p}.
The author was partially supported by Swiss National Science Foundation Fellowship No.\ 147304.}
}

\maketitle

\begin{abstract}
Spread codes and cyclic orbit codes are special families of constant dimension subspace codes. These codes have been well-studied for their error correction capability, transmission rate and decoding methods, but the question of how to encode and retrieve messages has not been investigated. In this work we show how a message set of consecutive integers can be encoded and retrieved for these two code families.
\end{abstract}

\begin{IEEEkeywords}
Message encoding, network coding, constant dimension codes, subspace codes, Grassmannian, enumerative coding, orbit codes, finite spreads, discrete logarithm.
\end{IEEEkeywords}

\section{Introduction}

Random network coding has received much attention in the last decade. \emph{Subspace codes}, first introduced in \cite{ko08}, are a class of codes used for random network coding. 
They are defined to be sets of subspaces of some given ambient space $\F_q^n$ of dimension $n$ over the finite field $\F_q$ with $q$ elements. When we restrict ourselves to subspace codes, whose codewords all have the same dimension $k$, we talk about \emph{constant dimension codes}.

A closely related area of research are rank-metric codes. These codes have already been studied before subspace codes, and it is known that one can construct optimal rank-metric codes, called maximum rank distance (MRD) codes, for any set of parameters. In \cite{de73,ga85a} a general construction for MRD codes was given. These codes are also called Gabidulin codes and they can be represented as a linear block code over some extension field of the underlying field.

Since the main idea of coding theory is to transmit information through a communication channel, any code should be able to encode information, or, in other words, \emph{encode messages} from a given message set. For generality, we will choose as message set the first non-negative integers $\mathcal{M}=\{0,1,2, \dots \}$. A message encoding map is an injective map from $\mathcal{M}$ to the code. The corresponding \emph{message retrieval} map is the inverse of the encoding map. From an application point of view it is very important that a code has an efficiently computable message encoding and corresponding retrieval map, since these need to be computed for every information transmission. 

In the seminal paper \cite{ko08} a class of Reed-Solomon-like constant dimension codes is proposed, which was later on shown to be equivalent to the lifting of Gabidulin codes \cite{si08j}. Due to the linearity of the Gabidulin code over an extension field, there exist efficient message encoding and retrieval maps for these codes, in analogy to the encoding and retrieval maps of linear block codes. 

During the last years other constructions of subspace codes were developed, e.g.\ in \cite{bo09,et09,et12,et11,ga11,ga10,gl14,go14,ko08p,ma08p,si15,sk10,tr11a,tr10}. Some of these constructions have the mere purpose of giving an improved transmission rate (i.e., larger cardinality of the code for the same parameters), while others have algebraic structure that can be used e.g.\ for decoding. 
The constructions from \cite{et09,si15} are based on the Reed-Solomon-like construction from \cite{ko08}, hence message encoding and retrieval can still be done based on the linearity of the underlying Gabidulin codes. 

However, for most of the other known subspace code constructions, the corresponding codes cannot be represented as a linear block code over some extension field, and it is not obvious how message encoding and message retrieval can be done for these codes.  
Surprisingly, this problem of message encoding and retrieval has received little attention in the before-mentioned and other related papers and will be the topic of this paper. We want to study this problem for two classes of subspace codes, namely \emph{spread codes} \cite{go14,ma08p} and \emph{orbit codes} \cite{gl14,tr11a,tr10}.
These two classes are of particular interest, since spread codes are optimal (see e.g.\ \cite{tr13phd}) with respect to their rate for a given error correction capability (and thus achieve a better rate than the codes from \cite{ko08}), and orbit codes have a lot of structure, which gives rise to code constructions and efficient error correcting decoding algorithms (see e.g.\ \cite{tr11a}).

The paper is organized as follows: In the following section we will give some preliminaries about finite fields and subspace codes, among others the definitions and constructions of spread codes and orbit codes. In Section \ref{sec:comprelim} we derive some preliminary results on computational complexities of tasks that we need later on in our message encoding and retrieval algorithms. In Section \ref{sec:spread} we derive an efficient encoding map for Desarguesian spread codes. In Section \ref{sec:orbit} we investigate message encoding for cyclic orbit codes. In Section \ref{sec:hybrid} we briefly describe how a message encoding and retrieval algorithm for a given subspace code can be combined with an error correcting decoding algorithm for another, semi-linearly isometric, code.
We conclude this work in Section \ref{conclusion}.


\section{Preliminaries}\label{sec:preliminaries}

In this section we give all the preliminaries we will need later on in the paper. We will first introduce finite fields and recall known results related to finite fields. Then we will do the same for subspace codes, where we also define spread and orbit codes. In the third subsection we define message encoding and retrieval maps and give a short overview of known results.

\subsection{Finite Fields}

In this subsection we recall some known facts about finite fields. The definitions and results can be found in any textbook on finite fields, e.g.\ in \cite{li94}.

Let $q$ be a prime power. We denote the finite field with $q$ elements by $\F_q$. 

\begin{defi}
A polynomial $p(x) \in \F_q[x]$ is called \emph{irreducible}, if it cannot be factored into the product of two non-constant polynomials of $\F_q[x]$. 
\end{defi}

\begin{lem}
\begin{enumerate}
\item If $\deg(p(x)) = k$, $p(x)$ is irreducible and $\alpha$ is a root of $p(x)$, then $\F_{q^k}\cong \F_q[\alpha]$.
\item If furthermore $\ord(\alpha)=q^k-1$, we say that $\alpha$ and $p(x)$ are \emph{primitive}.  In this case, $\F_{q^k}^* \cong \langle \alpha \rangle$. 
\end{enumerate}
 \end{lem}

\begin{lem}
The following map $\psi_k$ is a vector space isomorphism between $\F_q^k$ and $\F_q[\alpha]$:
\begin{align*}
\psi_k: \quad \F_q^{k} &\longrightarrow \F_{q}[\alpha]\\
(u_1,\dots, u_{k}) &\longmapsto \sum_{j=1}^{k} u_{j} \alpha^{j-1}
 \end{align*}
 \end{lem}
 
 \begin{defi}
The \emph{companion matrix} $P$ of some monic polynomial $p(x)=\sum_{i=0}^k p_i x^i$ is defined as
$$ P= \left( \begin{array}{ccccc} 
0 & 1& 0 & \dots &0 \\
0 & 0 & 1&  \dots &0 \\
\vdots &&& \ddots \\
0 & 0& 0 & \dots &1 \\
-p_0 & -p_1 & -p_2 &\dots & -p_{k-1}
 \end{array} \right).$$
 \end{defi}
Note that one often finds the transpose definition of a companion matrix in the literature.  However, in this work we will use the above row-wise definition. One can verify that multiplication with $P$, respectively $\alpha$, commutes with the vector space isomorphism $\psi_k$, which is stated in the following lemma. 

\begin{lem}
Let $p(x) \in \F_q[x]$ be monic, irreducible of degree $k$. Moreover, let $\alpha$ be a root of $p(x)$ and $P\in \F_q^{k\times k}$ the corresponding companion matrix. Then
$$ \psi_k(u)\alpha = \psi_k(uP)$$ for any $u\in\F_q^k$. This implies that $$\F_q[\alpha] \cong \F_q[P] .$$
\end{lem}
Throughout the paper we will denote by $\rho: \F_q[\alpha] \rightarrow  \F_q[P]$ the isomorphism given by $\rho(\alpha^i) = P^i$ and $\rho(0) = 0_{k\times k }$.

To set up our message encoding and retrieval maps later on, we need the following bijections between vector spaces over finite fields and integers sets. 
\begin{defi}
Let $p$ be prime and $m$ any positive integer. The map
\begin{align*}
\phi'_m : \quad \F_p^{m} &\longrightarrow \{0,1,\dots,p^m-1\}\\
(u_1,\dots, u_{m}) &\longmapsto \sum_{i=1}^{m} u_{i} p^{i-1}
\end{align*}
is called the inverse \emph{$p$-adic expansion}.
\end{defi}
It is well-known that, for a prime number $p$, $\phi'_m$ is a bijection. This map can be extended to a $q$-adic expansion, for a prime power $q=p^r$, by fixing a bijection $\varphi'$ between $\F_q=\F_{p^r}$ and $\{0,\dots,q-1\}$. To do so we represent $\F_{p^r}\cong \F_p[\beta]$ for a suitable $\beta$, and choose the bijection 
\begin{align*} 
\varphi'   : \quad  \F_{p}[\beta] &\longrightarrow \{0,1,\dots, q-1\} \\
\sum_{i=1}^{r} u_i \beta^{i-1} & \longmapsto \sum_{i=1}^{r} u_i p^{i-1}  .
\end{align*}
One can easily see that $\varphi'(1)=1$ and $\varphi'(0)=0$.

\begin{defi}
The inverse \emph{$q$-adic expansion} is given by
\begin{align*}
\phi''_m : \quad \F_q^{m} &\longrightarrow \{0,1,\dots,q^m-1\}\\
(u_1,\dots, u_{m}) &\longmapsto \sum_{i=1}^{m} \varphi'(u_{i}) q^{i-1}  .
\end{align*}
\end{defi}
To furthermore extend this map to a 
{$q^k$-adic expansion}, we need to fix a bijection $\varphi$ between $\F_{q^k}$ and $\{0,\dots,q^k-1\}$. To do so we represent $\F_{q^k}\cong \F_q[\alpha]$ for some suitable $\alpha$ and choose the bijection
\begin{align*} 
\varphi   : \quad  \F_{q}[\alpha] &\longrightarrow \{0,1,\dots, q^k-1\} \\
\sum_{i=1}^{k} u_i \alpha^{i-1} & \longmapsto \sum_{i=1}^{k} \varphi'(u_i) q^{i-1} .
\end{align*}
One can again easily see that $\varphi(1)=1$ and $\varphi(0)=0$.

\begin{defi}
The \emph{inverse $q^k$-adic expansion} is given by
\begin{align*}
\phi_{k,m} : \quad \F_{q^k}^{m} &\longrightarrow \{0,1,\dots,q^{km}-1\}\\
(u_1,\dots, u_{m}) &\longmapsto \sum_{i=1}^{m} \varphi(u_{i}) q^{k(i-1)}.
\end{align*}
\end{defi}
It can easily be verified that $\phi_{k,m}$ 
 is again a bijection. For computing the preimage $(u_1,\dots, u_{m})\in\F_{q^k}^m$ of some $j\in \{0,\dots,q^{km}-1\}$, i.e., the $q^k$-adic expansion of $j$, one recursively computes $\varphi(u_{\ell+1}) \equiv (i - \sum_{i=1}^{\ell}\varphi(u_{j}) q^{k(i-1)} )/q^{\ell k} \mod q^k$ with the initial congruence $\varphi(u_1) \equiv i \mod q^k$.



\subsection{Subspace Codes}

We denote the set of all subspaces of $\F_q^n$ by $\PG$ and the set of all subspaces of $\F_q^n$ of dimension $k$, called the \emph{Grassmannian}, by $\G$. We represent a vector space $\Uvs \in \G$ by a matrix $U\in \F_q^{k\times n}$ such that the row space of $U$, denoted by $\rs(U)$, is equal to $\Uvs$. 

\begin{defi}
A \emph{subspace code} is simply a subset of $\PG$ and a \emph{constant dimension code} is a subset of $\G$. 
\end{defi}

The following is a metric on $\PG$, and hence also on $\G$ (see e.g.\ \cite{ko08}). 
\begin{defi}
The \emph{subspace distance} is defined as
$$ d_S(\Uvs, \Vvs) := \dim(\Uvs) + \dim(\Vvs)- 2\dim(\Uvs\cap \Vvs) $$
for any $\Uvs, \Vvs \in \PG$. 
\end{defi}
The minimum distance $d_S(\Cvs)$ of a subspace code $\C \subseteq \PG$ is defined in the usual way, as the minimum of all pairwise distances of the codewords, i.e.,
$$d_S(\Cvs) :=\min \{d_S(\Uvs, \Vvs) \mid \Uvs, \Vvs \in \Cvs, \Uvs\neq \Vvs \} .$$ 
Since the dual of a subspace code $\C$ has the same minimum distance as $\C$ (see e.g.\ \cite{ko08}), it is customary to restrict oneself to $k\leq n/2$, which we will assume throughout the paper.

We will now introduce spread and orbit codes. These  families of constant dimension codes will be the focus of this paper.
\begin{defi}
A \emph{spread}, in $\G$ is defined as a set of elements of $\G$ that pairwise intersect only trivially and cover the whole space $\F_q^n$. 
\end{defi}
Spreads are well-known geometrical objects, see e.g.\ \cite{hi98}. Since spreads are subsets of $\G$, they can be used as constant dimension codes. In this case one also speaks of \emph{spread codes}, see e.g.\ \cite{ma08p}. 
The following properties of spread codes are well-known and can easily be derived.
\begin{lem} \cite{hi98,ma08p}
\begin{enumerate}
\item
Spreads in $\G$ exist if and only if $k| n$.
\item
A spread in $\G$ has minimum subspace distance $2k$ and cardinality $(q^n-1)/(q^k-1)$.
\item
A constant dimension code in $\G$ with minimum subspace distance $2k$ and cardinality $(q^n-1)/(q^k-1)$ is a spread.
\end{enumerate}
\end{lem}
For more information on different constructions and decoding algorithms of spread codes, see \cite{go12,ma08p,ma11j,tr13phd}. We will use the following well-known construction, which gives rise to a \emph{Desarguesian spread code} in $\G$ \cite{ba11a,tr13phd}.

\noindent \textbf{Construction I:}
\begin{itemize}
\item Let $\alpha$ be a root of an irreducible polynomial $p(x)\in \F_q[x]$ of degree $k$ and let $P$ be the corresponding companion matrix. Denote by $\rho: \F_q[\alpha]\rightarrow \F_q[P]$ the isomorphism from the previous subsection.

\item Represent $\F_{q^k}$ as $\F_q[\alpha]$. Let $m:=n/k$ and consider $\mathcal G_{q^k} (1,m)$, which has $q^{k(m-1)}+q^{k(m-2)} + q^{k(m-3)} + \dots + 1 = (q^n-1)/(q^k-1)$ elements. Naturally, all these lines intersect only trivially.
\item Define the map 
\begin{align*}
\mathrm{des} : \quad \mathcal G_{q^k} (1,m) &\longrightarrow \G \\
 \rs(v_1,v_2,\dots,v_n) &\longmapsto \rs(\rho(v_1),\rho(v_2) \dots, \rho(v_m))  .
\end{align*}
Then the image of $\mathrm{des}$ is a Desarguesian spread in $\G$.
\end{itemize}

Note that the name \emph{Desarguesian} arises from the fact that the translation planes of these spreads are Desarguesian planes. For our purposes though, this fact is not needed, we simply use the name for the construction from above.

In our message encoding algorithms for spread codes we need to have a unique description of the elements of $\mathcal{G}_{q^{k}}(1,m)$. To do so we choose the normalized basis vector, i.e., the one element of the one-dimensional subspace whose first non-zero entry is equal to one. This normalized vector is then mapped by $\mathrm{des}$ to the reduced row echelon form of the respective codeword. 
%
The reader familiar with projective spaces will notice that $\mathcal G_{q^k} (1,m)$ corresponds exactly to the projective space over $\F_{q^k}$ of dimension $m-1$. The usage of normalized representatives of points is a common concept there.

\begin{ex}\label{ex1}
Let $q=k=m=2, n=4$ and $\alpha$ be a root of $x^2+x+1$, i.e., a  primitive element of $\F_{2^2}\cong \F_2[\alpha]$. The respective companion matrix is 
$$P=\left(\begin{array}{cc} 0&1\\1&1 \end{array}\right) .$$
Then $\mathcal G_{2^2} (1,2) = \{ \rs (1 , 0), \rs(1 , \alpha), \rs (1 , \alpha^2), \rs(1 , 1), \rs(0 , 1) \}$ and substituting all elements of $\F_2[\alpha]$ with its corresponding element from $\F_2[P]$ gives a spread in $\mathcal G_2(2,4)$:
\[\left\{\rs\left(\begin{array}{cccc}1&0&0&0\\0&1&0&0 \end{array} \right), 
\rs\left(\begin{array}{cccc}1&0&0&1\\0&1&1&1 \end{array} \right), 
\rs\left(\begin{array}{cccc}1&0&1&1\\0&1&1&0 \end{array} \right), 
\rs\left(\begin{array}{cccc}1&0&1&0\\0&1&0&1 \end{array} \right), \rs\left(\begin{array}{cccc}0&0&1&0\\0&0&0&1 \end{array} \right)  \right\}\]
\end{ex}

\vspace{0.5cm}

\emph{Orbit codes} \cite{tr10p} in $\G$ are defined to be orbits of a subgroup of the general linear group $\GL_n:=\{A \in \F_q^{n\times n} \mid \rk(A)=n \}$ of order $n$ over $\F_q$:
\begin{defi}
Let $\Uvs \in \G$ and $G$ be a subgroup of $\GL_{n}$. Then 
$$\Uvs G = \{\Uvs A \mid A\in G \}$$
is called the \emph{orbit code} generated by the initial point $\Uvs$ and the group $G$.
\end{defi}
As shown in \cite{tr11a}, orbit codes can be seen as the analog of linear codes in classical block coding. Their structure can be used for an easy computation of the minimum distance of a code, as well as for decoding algorithms (e.g.\ one can coset-leader decode them). For more information on orbit codes the interested reader is referred to \cite{ma11p,ro12j,tr13phd,tr11a}. 

\begin{ex}
Let $q=k=m=2$ and $n=3$. Moreover, let 
\[\Uvs = \rs\left(\begin{array}{cccc} 0&1&0\\0&0&1 \end{array}\right)\]
and $G$ be the cyclic group generated by the matrix
\[P= \left(\begin{array}{cccc} 0&1&0  \\ 1&0&0\\ 0&0&1 \end{array}\right) .\]
The corresponding orbit code 
$$\Uvs G =\{\Uvs P^i \mid i=0,\dots,|G|-1\}  = \left\{\rs\left(\begin{array}{cccc} 0&1&0\\0&0&1 \end{array}\right), \rs\left(\begin{array}{cccc} 1&0&0\\0&0&1 \end{array}\right) \right\}\subset \mathcal G_2(2,3)$$ 
has  $2$ elements and minimum subspace distance $2$. 
\end{ex}

One can also use the orbit code construction to construct spread codes, as illustrated in the following example.

\begin{ex}\label{ex2}
Let $q=k=m=2$ and $n=4$. Moreover, let 
\[\Uvs = \rs\left(\begin{array}{cccc} 1&0&0&0\\0&1&1&0 \end{array}\right)\]
and
\[P = \left(\begin{array}{cccc} 0&1&0&0 \\ 0&0&1&0 \\ 0&0&0&1 \\ 1&1&0&0 \end{array}\right)\]
be the companion matrix of the irreducible polynomial $x^4+x+1 \in \F_2[x]$. 
The  group $G=\langle P \rangle$ is a subgroup of $\GL_4$ of cardinality $15$. The corresponding orbit code $\Uvs G =\{\Uvs P^i \mid i=0,\dots,14\} \subset \mathcal G_2(2,4)$ has  $5$ elements and minimum subspace distance $4$. 
Hence, it is a spread code in $\mathcal G_2(2,4)$.
\end{ex}

The following lemma gives a general construction of a spread in $\G$ as a cyclic orbit code. This construction is again well-known and can be found e.g.\ in \cite{ba11a,tr11a}.
\begin{lem}
Let $P\in \GL_n$ be a companion matrix of a monic primitive polynomial $p(x) \in \F_q[x]$ of degree $n$. Moreover, let $\Uvs\in \G$ be the vector space representation of the subfield $\F_{q^k}$ of $\F_{q^n}$. Then $\Cvs=\Uvs \langle P\rangle$ is a spread code in $\G$.
\end{lem}


\subsection{Message encoding and retrieval}
We can define message encoding and retrieval maps very general, for any type of code, as follows.

\begin{defi}
For a given code $\C$ in some space $X$ and some message space $\cM$, an \emph{encoding map} for the code $\Cvs$ 
$$\enc : \cM \longrightarrow X$$
is an injective map with image $\Cvs$. 
The inverse map 
$$\enc^{-1} : \C \longrightarrow \cM $$
is called the corresponding \emph{message retrieval map}.
\end{defi}

In our setting of subspace codes, $X=\PG$, or if we use only constant dimension codes, $X=\G$. 

Mostly in the information theory literature, a general message set is represented as $\cM=\{0,\dots,|\C|-1\}$. For classical linear block codes in $\F_q^n$ the usual message space is $\cM= \F_q^k$ for some integer $k\leq n$. With the $q$-adic expansion this can easily be translated to the message set $\{ 0, \dots, q^k-1\}$. In contrast, not any set of integers $\{0,\dots,j-1\}$ can be bijectively mapped to some linear vector space. In particular, if $j$ is not a prime power, there exists no linear vector space over a finite field of the same cardinality. In this paper, since the cardinalities of our codes are in general not prime powers, we derive encoding maps for  message sets of the form $\cM=\{0,\dots,|\C|-1\}$.

In the subspace coding case it is not obvious in general, how message encoding or message retrieval can be done. However, an elegant solution is given for the Reed-Solomon-like codes in \cite{ko08}. For such a code $\C \subseteq \G$ the message space is
$$\cM=\F_{q^{n-k}}^{k-\frac{d_S(\C)}{2}+1} ,$$
which is isomorphic (as a vector space) to $\F_{q}^{(n-k)(k-d_S(\C)/2+1)}$, and an encoding map for $\C$ is given by
\begin{align*}
 \enc : \quad \F_{q^{n-k}}^{k-\frac{d_S(\C)}{2}+1}  &\longrightarrow  \G \\
 (u_1,\dots, u_{k-\frac{d_S(\C)}{2}+1}) &\longmapsto \langle  (\psi^{-1}_{k}(\beta_j), \psi^{-1}_{n-k} (\sum_{i=0}^{k-\frac{d_S(\C)}{2}} u_{i+1} \beta_j^{q^i} ) )  \mid j=1,\dots, k\rangle
\end{align*}
where $\beta_1,\dots, \beta_k \in \F_{q^{n-k}}$ are linearly independent over $\F_q$, and we use the isomorphisms $\langle\beta_1,\dots, \beta_k\rangle \cong \F_q^{k}$ and $\F_{q^{n-k}}\cong \F_q^{n-k}$ for the two vector entries on the right side, respectively. 
Via interpolation this map is invertible and the inverse is computable in polynomial time. Hence, one gets a feasible message retrieval map as well. In fact, in the decoding algorithm of \cite{ko08}, error correction and message retrieval is done in one algorithm.

As already mentioned in the introduction, many subspace codes cannot be represented as a linear block code over some extension field, hence the idea from above is not necessarily adaptable to other subspace code constructions. Therefore we will use other, different approaches to derive encoding and retrieval maps for the two classes of codes we will investigate in this paper. 

One of the ideas we will pursue, is to use \emph{enumerative coding} for message encoding for constant dimension codes. Enumerative coding for the Grassmannian was studied in \cite{si11}, where the idea of enumerative source encoding of $q$-ary block codes of Cover \cite{co73} was translated to a subspace setting. In contrast to our contribution, the algorithms of \cite{si11} are only stated for the whole Grassmannian, and not for any error-correcting constant dimension codes. The idea of enumerative coding was also used in \cite{sc14} to encode subspace Gray codes. These codes are, however, no-error-correcting. To adapt the ideas of enumerative coding to spread or orbit codes, one needs an efficiently computable map that counts the number of subspaces whose reduced row echelon form fulfills certain requirements. This is feasible for Desarguesian spread codes and will be explained in Subsection \ref{sec:enuDes}. 
 For orbit codes however, it is not clear how such a map could efficiently be computed, which is why we will not pursue the idea of enumerative coding as an encoding map in this context.


\section{Computational Preliminary Results}\label{sec:comprelim}

In this section we derive complexity orders of tasks we will need in our main algorithms in Sections \ref{sec:spread} and \ref{sec:orbit}. 
For comparability with error decoding complexities we will do our complexity analyses over $\F_q$, which is why we will represent the messages $0,1,\dots, |\Cvs|-1$ in their $q$-adic expansion. For simplicity we represent these $q$-adic expansions in $\F_q^n$, although not all coordinates are necessarily needed. 

We use the Big-O notation for the computational complexities of our algorithms, where we use the index $q$ to specify that the given complexity order is over the base field $\F_q$ and not over some extension field. 

The following results are well-known.
\begin{lem}\label{lem:compbasic}
Consider $\F_q$ and an extension field $\F_{q^k}$. Represent $\F_{q^k}\cong \F_q[\alpha]$ for some suitable $\alpha$.
\begin{enumerate}
\item Multiplying two elements from $\F_{q^k}$ can be done with $\mathcal{O}_q(k^2)$ operations in $\F_q$. The same holds for division in $\F_{q^k}$.
\item Let $\beta \in \F_{q^k}$ and $0\leq i \leq q^k-1$. Computing the modular exponentiation $\beta^i$ of $\beta \in \F_{q^k}$, i.e., finding the representation of $\beta^i$ in the basis $\{1, \alpha, . . . , \alpha^k-1\}$ of $\F_{q^k}$, can be done with at most $\mathcal{O}_q(k^3)$ operations in $\F_q$.
\end{enumerate}
\end{lem}
\begin{proof}
\begin{enumerate}
\item Any element in $\F_{q^k}$ can be represented as a polynomial over $\F_q$ of degree less than $k$. Since polynomial multiplication and division of polynomials of degree at most $k$ can be done with $\mathcal{O}_q(k^2)$ operations (see e.g.\ \cite[Corollary 4.6]{ga03}), the first statement follows.
\item Using a normal basis of $\F_{q^k}$, it was shown in \cite{ga91} that modular exponentiation can be done with $k/( \log_q k)$ multiplications in $\F_{q^k}$. The change of basis to the normal basis is a linear map and can hence be done with  $\mathcal{O}_q(k^2)$ operations in $\F_q$. Together with 1) the statement follows.
\end{enumerate}
\end{proof}

\begin{lem}\label{lem:qadic}
Let $\alpha$ be a root of a monic irreducible polynomial $p(x) \in \F_q[x]$ of degree $k$, such that $\F_{q^k}\cong \F_q[\alpha]$. 
The complexity of computing the map $\psi_{k}: \F_{q}^k \rightarrow \F_q[\alpha]$, as well as computing its inverse, is in $\mathcal O_q(k)$. 
\end{lem}
\begin{proof}
Let $u=(u_1,u_2,\dots,u_k)\in \F_q^k$. If we want to compute $\psi(u)$ we simply need to write the $k$ vector coordinates $u_1,\dots,u_k$ as polynomial coefficients in $\sum_{i=1}^k u_i \alpha^{i-1}$. The inverse map can be computed analogously, by writing the polynomial coefficients as vector entries. Therefore both maps can be computed with a complexity in $\mathcal{O}_q(k)$.
\end{proof}


\begin{lem}\label{lem:compphi}
The complexity of computing the map $\phi_{k,m}: \F_{q^k}^m \rightarrow \{0,1,\dots, q^{km}-1\}$, as well as computing its inverse map $\phi_{k,m}^{-1}$, 
is in $\mathcal O_q(km)=\mathcal O_q(n)$. 
\end{lem}
\begin{proof}
Recall that $\phi_{k,m}(u_1,u_2,\dots, u_m) = \sum_{i=1}^{m} \varphi(u_{i})q^{k(i-1)} $ and that we represent the integers in their $q$-adic expansion in $\F_q^n$. Denote by $(i_1,i_2,\dots,i_n)$ the $q$-adic expansion of the integer $i$. Then $ \phi_{k,m}^{-1}(i_1,i_2,\dots,i_n) = (\psi_k(i_1,i_2,\dots,i_k), \dots, \psi_k(i_{n-k+1},i_{n-k+2},\dots, i_n))$. 
It follows from Lemma \ref{lem:qadic} that the complexity of computing $ \phi_{k,m}^{-1}$ is in $\mathcal O_q(mk)=\mathcal O_q(n)$.

Similarly, we get $\phi_{k,m}(u_1,u_2,\dots, u_m) = (\psi_k^{-1}(u_1),\psi_k^{-1}(u_2),\dots, \psi_k^{-1}(u_m))$, which is the respective integer in its $q$-adic expansion. 
By Lemma \ref{lem:qadic} this can again be done with a complexity in $\mathcal O_q(n)$.
\end{proof}

The next task we will investigate is computing powers of companion matrices of irreducible polynomials.

\begin{lem}\label{lemPexp}
 Let $p(x)\in\F_q[x]$ be monic irreducible of degree $k$, $\alpha$ a root of $p(x)$,  and let $P\in\GL_k$ be its companion matrix. Then $P^i$ can be computed with a complexity of at most $\mathcal{O}_q(k^3)$ operations over $\F_q$. 
 
 If moreover 
  the representation of $\alpha^i$ in the basis $\{1,\alpha,\dots, \alpha^{k-1}\}$ of $\F_q[\alpha]$ is known, then $P^i$ can be computed with a complexity in  $\mathcal{O}_q(k^2)$.
\end{lem}
\begin{proof}
 Recall from Section \ref{sec:preliminaries} that $\psi_k(uP)=\psi_k(u)\alpha$ for any $u\in\F_q^k$. Thus, if we apply $\psi_k$ on the rows of $P^i$, we get
\[\psi_k(P^i)=\psi_k(P)\alpha^{i-1}=
\left(\begin{array}{ccccc}
                   \alpha\\ \alpha^{2}\\ \alpha^{3}\\\vdots\\ \alpha^{k}  
                  \end{array}\right)\alpha^{i-1} =
                  \left(\begin{array}{ccccc}
                   \alpha^i\\ \alpha^{i+1}\\ \alpha^{i+2}\\\vdots\\ \alpha^{i+k-1}  
                  \end{array}\right).
\]
For $0\leq i \leq k-1$ we have $\psi_k^{-1}(\alpha^i)=e_{i+1}$, where $e_i\in \F_q^k$ is the $i$-th unit vector. For higher values of $i$ we need to compute the  representation of $\alpha^i$ in the basis $\{1,\alpha,\dots, \alpha^{k-1}\}$. 
This representation, if not known, can be computed with a complexity of at most $\mathcal{O}_q(k^3)$ (see Lemma \ref{lem:compbasic}).  

Then we can construct $P^i$ as follows: 
The first row is simply the vector representation $\psi_k^{-1}(\alpha^i) \in \F_q^k$ of 
 $\alpha^i \in \F_{q}[\alpha]$. This can be done with $k$ coefficient transfers (over $\F_q$). 
For $2\leq j\leq k$ the $j$-th row of $P^i$, denoted by $P_{j}^{i}$, is  given by
\[P_{j}^{i} = P^i_{j-1} P = (\:0 \; , \; (P_{j-1}^{i})_{[1,k-1]}\: ) + (P_{j-1}^{i})_{k} \cdot (-p_0, -p_1,\dots, -p_{k-1} )  ,\]
where $(P_{j-1}^{i})_{[1,k-1]}$ denotes the subvector of  $P_{j-1}^{i}$ without the last coordinate and $(P_{j-1}^{i})_{k}$ denotes the last coordinate of $P_{j-1}^{i}$. Hence, for each row of $P^i$ we need to multiply a vector in $\F_q^k$ by a scalar from $\F_q$ and then add two vectors from $\F_q^k$. Both of these computations need $k$ operations in $\F_q$. Since we need to do this for each row of $P^i$, we get an overall complexity order of $\mathcal{O}_q(k^2)$, if  the representation of $\alpha^i$ in the basis $\{1,\alpha,\dots, \alpha^{k-1}\}$ is known. 
If this representation of $\alpha^i$ is not known, the overall complexity becomes $\mathcal{O}_q(k^3)$.
\end{proof}

We can now derive the complexity of computing the map $\mathrm{des}$, for constructing a Desarguesian spread code, as defined in Construction I in Section \ref{sec:preliminaries}:
\begin{thm}\label{thm1}
Consider the map $\mathrm{des}:\mathcal G_{q^k} (1,m) \rightarrow \G$, whose image is a Desarguesian spread $\C \subseteq \G$. 
The map $\mathrm{des} $ and its inverse $\mathrm{des}^{-1} :\C \rightarrow \mathcal G_{q^k} (1,m)$ can be computed with a complexity order in $\mathcal{O}_q (kn)$.
\end{thm}
\begin{proof}
As before, let $p(x) \in \F_q[x]$ be monic irreducible of degree $k$, $\alpha$ be a root of $p(x)$, $P$ the corresponding companion matrix and let $\rho$ be the isomorphism from $\F_q[\alpha]$ to $\F_q[P]$.
For the computation of the map $\mathrm{des}$, take the normalized representation of the preimage $(u_1,\dots,u_m)\in\mathcal G_{q^k} (1,m)$ and consider the elements $u_i\in\F_{q^k} \cong \F_q[\alpha]$. 
For each $i\in \{1,\dots,m\}$ such that $u_i\neq 0$, use the construction used in the proof of Lemma \ref{lemPexp} to construct $P_i = \rho( u_i)$ for $i=1,2,\dots,m$. 
By Lemma  \ref{lemPexp}, this can be done with a complexity in $\mathcal{O}_q(k^2)$.
Then 
$$ \rs (\rho(u_1),\rho(u_2),\dots, \rho(u_m)) =  \rs (P_1,P_2, \dots, P_m) \in \G$$
is the respective spread codeword. Since we need to construct at most $m = n/k$ matrices $P_i$, we get an overall complexity order of $\mathcal{O}_q(mk^2)=\mathcal{O}_q(kn)$.

We now consider the inverse map $\mathrm{des}^{-1}$. Choose one vector $v\in \F_q^n$ of the given codeword $\Uvs \in \G$ and represent it as $u\in \F_{q^k}^m$ via
$$ u = (\psi_k(v_1,v_2,\dots, v_k), \psi_k(v_{k+1},v_{k+2},\dots, v_{2k}), \dots,\psi_k(v_{n-k+1},v_{n-k+2},\dots, v_n)) .$$ 
By Lemma \ref{lem:qadic} this representation can be done with $\mathcal{O}_q(mk)=\mathcal{O}_q(n)$ operations. 
Normalize $u$ by dividing all coordinates by the first non-zero entry of $u$. This normalized vector is the representative of the respective element in $\mathcal G_{q^k} (1,m) $. For the normalization, one needs at most $m$ divisions over $\F_{q^k}$. Each such division can be done with $\mathcal{O}_q (k^2)$ operations (see Lemma \ref{lem:compbasic}), i.e., we get an overall complexity of $\mathcal{O}_q ( mk^2)=\mathcal{O}_q (kn)$.
\end{proof}


\section{Message Encoding for Desarguesian Spread Codes}\label{sec:spread}

In this section we derive message encoding and retrieval maps for Desarguesian spread codes. In the first subsection we derive an intuitive encoder for these type of codes, arising from the Grassmannian representation in Construction I. In the second subsection we use the idea of enumerative coding on the Grassmannian to derive an encoder for Desarguesian spread codes. We then show that this second encoder is the same as the encoder from the first subsection with a little twist.


\subsection{Ad Hoc Construction}\label{ssec:adhoc}

We will now derive a message encoding map by concatenating the map $\mathrm{des}$ with an injective map $f$ from $\{0,\dots, (q^n-1)/(q^k-1) -1\}$ to $\mathcal{G}_{q^{k}}(1,m)$. This map is defined as follows:
\begin{align*}
 f : \{0,\dots, (q^n-1)/(q^k-1)-1\} &\longrightarrow \mathcal G_{q^k} (1,m)   \\
  i &\longmapsto  \rs (\underbrace{0, \dots ,0 }_{m-\epsilon(i)-1} ,1,  \phi_{k,\epsilon(i)}^{-1}(i- \sum_{j=0}^{\epsilon(i) -1} q^{jk}) )   .
\end{align*}
where $\epsilon(i) := \min\{\ell \mid \sum_{j=0}^{\ell} q^{jk} \geq i+1\}$ and $\phi_{k,\epsilon(i)} : \F_{q^k}^{\epsilon(i) } \rightarrow\{0,\dots, q^{k\epsilon(i)}-1\} $ is the inverse $q^k$-adic expansion, as explained in Section \ref{sec:preliminaries}.
We defined $\epsilon(i)$ such that $f$ behaves as follows: 
\begin{align*}
0 &\mapsto\; \rs(0,\dots,0,0,0,1) ,\\
\{1,2,\dots,q^k\}\ni i &\mapsto \{\rs(0,\dots,0,0,1, v) \mid v\in \F_{q^k}, v=\phi_{k,1}^{-1}(i-1)\} , \\
\{q^k+1,q^k+2,\dots,q^{2k}+q^k\}\ni i &\mapsto \{\rs(0,\dots,0,1,v) \mid v\in  \F_{q^k}^2, v=\phi_{k,2}^{-1}(i-q^k-1)\} ,  \\
\vdots\\ 
\Big\{\sum_{j=1}^{m-2} q^{jk} +1,\sum_{j=1}^{m-2} q^{jk} +2,\dots,\sum_{j=1}^{m -1} q^{jk}   \Big\} \ni i &\mapsto \Big\{\rs(1, v)  \mid v \in \F_{q^k}^{m-1} ,   v=\phi_{k,m-1}^{-1}(i-\sum_{j=0}^{m-2} q^{jk}) \Big\}.
\end{align*}

\begin{thm}\label{thm:enc2}
The map $f$ is bijective and hence
$$\enc_{1} := \mathrm{des} \circ f$$
is an injective map from $\{0,\dots, (q^n-1)/(q^k-1)-1 \}$ to $\G$, whose image is the Desarguesian spread code $\C\subseteq \G$ from Construction I. Therefore, $\enc_{1}$ is an encoding map for the respective Desarguesian spread code $\C$.
\end{thm}
\begin{proof}
By the above shown behavior of $f$ and the fact that $\phi_{k,\epsilon(i)}$ is bijective, it follows that $f$ is injective. Since $(q^n-1)/(q^k-1) =\sum_{j=0}^{m -1} q^{jk}$, the cardinalities of domain and codomain of $f$ are equal. This implies that $f$ is bijective. Since the image of  $\mathrm{des}$ is the Desarguesian spread code $\C\subseteq \G$, the statement follows. 
\end{proof}

\vspace{0.5cm}

We will now give two algorithms, describing how to compute the encoding map $\enc_{1}$ and the respective message retrieval map $\enc_{1}^{-1}$. As before we denote by $\rho$ the isomorphisms from $\F_{q}[\alpha]$ to $\F_{q}[P]$. 
The computational complexity order of these two algorithms is afterwards given in Theorem \ref{thm:DesComp}.


\begin{algorithm}[ht]
\begin{algorithmic}
\REQUIRE{A message $i\in \{0,1,\dots, (q^n-1)/(q^{k}-1)-1\}$.}
\STATE{Compute $\epsilon(i)$.} 
\STATE{Compute $i' = i- \sum_{j=0}^{\epsilon(i) -1} q^{jk}$.}
\STATE{Compute $u'=\phi^{-1}_{k,\epsilon(i)} (i') $.}
\STATE{Set $u:=(\underbrace{0, \dots ,0 }_{m-\epsilon(i)-1} ,1, u')$.}
\STATE{Set $U:=(\rho(u_{1}),\rho(u_{2}), \dots, \rho(u_{m}))$.}
\RETURN $\Uvs=\rs (U)$
\end{algorithmic}
\caption{Message encoding for the Desarguesian spread code  $\C \subseteq \G$  from Construction I.}
\label{alg:DesEnc}
\end{algorithm}

\begin{algorithm}[ht]
\begin{algorithmic}
\REQUIRE{A spread codeword $\Uvs\in \C$.}
\STATE{Choose a non-zero vector $v\in \Uvs$.} 
\STATE{Compute  $u=(\psi_k(v_1,v_2,\dots,v_k),\psi_k(v_{k+1},v_{k+2},\dots,v_{2k}), \dots, \psi_k(v_{n-k+1},v_{n-k+2},\dots,v_n))$.}
\STATE{Normalize $u$.}
\STATE{Set $\epsilon(i):= m -$ (the coordinate of the first non-zero entry of $u$).}
\STATE{Set  $u':= $ the rightmost $\epsilon(i)$ coordinates of $u$.}
\STATE{Compute $i=\phi_{k,\epsilon(i)} (u') + \sum_{j=0}^{\epsilon(i)-1} q^{jk}$.}
\RETURN $i$
\end{algorithmic}
\caption{Message retrieval for the Desarguesian spread code  $\C \subseteq \G$  from Construction I.}
\label{alg:Des}
\end{algorithm}

\begin{thm}
Algorithms \ref{alg:DesEnc} and \ref{alg:Des} compute the images of the maps $\enc_1$ and $\enc_1^{-1}$, respectively, for any valid input.
\end{thm}
\begin{proof}
Algorithm \ref{alg:DesEnc} first computes the image $f(i)$ in the first three steps and then computes the image of this result under the map $\mathrm{des}$. Hence, by Theorem \ref{thm:enc2}, Algorithm \ref{alg:DesEnc} returns the corresponding Desarguesian spread codeword to the input.

In Algorithm \ref{alg:Des} we note that the choice of the non-zero vector $v\in \Cvs$ does not matter, since any choice will result in the same normalized $u\in\F_{q^k}^m$. One can easily check that the last three steps of the algorithm then compute $f^{-1}(u)$. For this note that in the definition of $f$ one can see that $\epsilon(i)$ is equal to $m$ minus the coordinate of the first non-zero entry of $u$. Hence, by Theorem \ref{thm:enc2}, Algorithm \ref{alg:Des} returns the message corresponding to a codeword of the Desarguesian spread $\Cvs$.
\end{proof}

\begin{ex}\label{ex6}
Let $q=k=2$  and $m=3$. Moreover, let $\alpha\in \F_{2^2}$ be a primitive element and let the elements of $\F_{2^2}$ be identified via the map $\varphi$ with $0\mapsto 0, 1\mapsto 1, \alpha\mapsto 2, \alpha^2=\alpha+1 \mapsto 3$. 
Moreover, let $P$ be the companion matrix of $\alpha$ (see Example \ref{ex1}).
We want to encode the message $i=14$. Following Algorithm \ref{alg:DesEnc}, we compute $\epsilon(i)=2$, $i'= 14-(1+4)=9$ and $u'=\phi_{2,2}^{-1}(9)=(1,\alpha)$. Thus we get  $u=(1,1,\alpha)$ and $U=(I_2, I_2, P)$ as the basis matrix of the respective codeword. 

Similarly, we get that the first $10$ non-negative integers are mapped by $f$ to the following elements of $\mathcal G_{2^2}(1,3)$:
$$\begin{array}{|c|c|c|c|c|c|c|c|c|c|c|c|}
\hline
i   & 0&  1&2&3&4&5&6&7&8&9&...\\\hline
f(i)& \rs(0, 0, 1)   &\rs(0, 1, 0)   &\rs(0, 1, 1)   &\rs(0, 1, \alpha)   &\rs(0, 1, \alpha^2)   &\rs(1, 0, 0)   &\rs(1, 1, 0)   &\rs(1, \alpha, 0)   &\rs(1, \alpha^2, 0)   &\rs(1, 0,1)   &... \\\hline
\end{array}
$$
Following Algorithm \ref{alg:DesEnc}, the matrix representation of the spread code elements in $\mathcal G_{2}(2,6)$ are given by replacing the elements of $\F_{2}[\alpha]$ by the respective element of $\F_2[P]$.
\end{ex}

\begin{ex}
Consider the same setting as in Example \ref{ex6}. Let $\Uvs= \rs(\rho(1), \rho(\alpha + 1), \rho(1)) \in \Cvs$ be a codeword, for which we would like to find the corresponding message. Following Algorithm \ref{alg:Des}, we choose some non-zero $v\in \Uvs$, say $v=(0, 1 , 1, 1, 0, 1)$, and compute $u=(\alpha,1,\alpha)$. Then we normalize $u$ to $(1,\alpha+1, 1)$. Since the first non-zero entry is in position $1$, we get $\epsilon(i)=3-1=2$. Then we compute $i= \phi_{2,2}(\alpha+1,1) + \sum_{j=0}^{1} 2^{2j} = (3\cdot 1 +1\cdot 4) +(1+4) =12$.
\end{ex}

In the following we analyze the complexity of computing $\enc_1$ and $\enc_1^{-1}$, i.e., Algorithms \ref{alg:DesEnc} and \ref{alg:Des}. 
For this we need the following lemma, which implies that the computation of $\epsilon(i)$ can be done efficiently in the $q$-adic expansion. 

\begin{lem}\label{lem:epsilon}
For any $\ell<m$, the $q$-adic expansion $ ( u_1, \dots, u_n)\in\F_q^n$ of the integer $\sum_{j=0}^{\ell} q^{jk} $ is given by
$$\phi_{1,n}^{-1} \left(\sum_{j=0}^{\ell} q^{jk}\right)= ( u_1, \dots, u_n)   ,\quad \textnormal{ where }\left\{ \begin{array}{ll} u_{i} = 1 &  \textnormal{if } k| (i-1) \textnormal{ and } i-1\leq  \ell k \\ u_{i}=0 & \textnormal{else } \end{array}\right.   .$$
\end{lem}
\begin{proof}
The inverse $q$-adic expansion maps $(u_1,\dots,u_n)$ to $\sum_{i=1}^{n} \varphi(u_{i}) q^{i-1}$. 
With the above values of $u_1,\dots, u_n$ one gets (recall that $\varphi(0)=0$ and $\varphi(1)=1$) 
$$\phi_{1,n}(u_1,u_2,\dots,u_n)= \sum_{i=1}^{n} \varphi(u_{i}) q^{i-1}= q^0+q^k +q^{2k}+\dots+q^{\ell k}=\sum_{j=0}^{\ell} q^{jk}    .$$
\end{proof}

To check if an integer $i$ is greater than an integer $j$, one needs to check if the $q$-adic expansion of $i$ is greater than the $q$-adic expansion of $j$ in reverse lexicographic order.

\begin{ex}
Consider the setting of Example \ref{ex6}. We represent the message set $\{0,1,\dots,20\}$ in their $2$-adic expansion:
$$\begin{array}{l l l l }
 0 \rightarrow (0,0,0,0,0,0)   &   2 \rightarrow (0,1,0,0,0,0) & \quad \dots \quad & 19 \rightarrow (1,1,0,0,1,0)\\ 
 1 \rightarrow (1,0,0,0,0,0)   &  3 \rightarrow (1,1,0,0,0,0) &  & 20 \rightarrow (0,0,1,0,1,0)   .\\ 
\end{array}
$$
By Lemma \ref{lem:epsilon} the $2$-adic expansion of $\sum_{j=0}^{\ell} 2^{2j}$ for $\ell \in \{0,1,2\}$ are given by:
$$\begin{array}{l l l l }
\ell=0: & (1,0,0,0,0,0) \\
\ell=1: & (1,0,1,0,0,0) \\
\ell=2: & (1,0,1,0,1,0)  . 
\end{array} $$
Hence, we can compute $\epsilon(i)=\min\{\ell \mid \sum_{j=0}^{\ell} 2^{jk} \geq i+1\}=\min\{\ell \mid \sum_{j=0}^{\ell} 2^{jk}-1 \geq i\}$ for a given message $i\in \{0,1,\dots,20\}$ in its $2$-adic expansion $u= (u_1,\dots, u_6) \in \F_2^6$ as follows (note that $u_6=0$ for all messages):
\begin{itemize}
 \item If $u=0$, then $\epsilon(i)=0$.
 \item If $u\neq 0$, 
check coordinate-wise (from right to left) if $u$ is less than or equal to $(0,0,1,0,0,0)$. If so, then $\epsilon(i)=1$, otherwise $\epsilon(i)=2$. 
\end{itemize}
\end{ex}

In the proof of the following lemma we describe how to 
compute $\epsilon(i)$ in general.
\begin{lem}\label{lem30}
The complexity of computing $\epsilon(i)$, for $i\in \{0,1,\dots, q^n-1\}$, is in $\mathcal{O}_q(n)$.
\end{lem}
\begin{proof}
If $i=0$, then $\epsilon(i)=0$. Assume now that $i>0$.
Represent $i$ in its $q$-adic expansion  $\phi_{1,n}^{-1}(i)= (u_1,\dots, u_n)=u \in \F_q^n$.
Find the first non-zero coordinate from the right ${j^*}$ of $u$, i.e., $j^*=\min\{i \mid u_{n-i}=0 \}$. 
\begin{itemize}
\item
If $k\nmid (j^*-1) $, then $\epsilon(i) = \lceil ({j^*-1})/{k}\rceil $. 
\item 
If $k| (j^*-1)$, then compare $u$ coordinate-wise (from right to left) with the $q$-adic expansion of $\sum_{j=0}^{(j^*-1)/k} q^{jk}-1$, as illustrated in Lemma \ref{lem:epsilon}. 
If $u$ is strictly greater than the $q$-adic expansion of $\sum_{j=0}^{(j^*-1)/k} q^{jk}-1$, then $\epsilon(i) = (j^*-1)/k +1$, otherwise $\epsilon(i) = (j^*-1)/k$. 
\end{itemize}
Thus we need at most $n$ coordinate comparisons, and a division over $\F_q$, which implies the statement.
\end{proof}

We can now give the computational complexity order of the message encoding and retrieval for Desarguesian spread codes:

\begin{thm}\label{thm:DesComp}
Algorithms \ref{alg:DesEnc} and \ref{alg:Des} have a computational complexity in $\mathcal O_q(kn)$.
\end{thm}
\begin{proof}
For the encoder, i.e., Algorithm \ref{alg:DesEnc}, the complexity of computing $f(i)$ is dominated by finding $\epsilon(i)$ and computing $\phi_{k,\epsilon(i)}^{-1}(i- \sum_{j=0}^{\epsilon(i) -1} q^{jk})$. By  Lemmas~\ref{lem:qadic} and \ref{lem30} these task can be done with $\mathcal O_q(n)$ operations. 
Since we know from Theorem \ref{thm1} that $\mathrm{des}$ can be computed with a computational complexity of order $\mathcal O_q(kn)$, the statement for $\enc_1$ follows.

Algorithm~\ref{alg:Des} describes how to compute the retrieval map $\enc_1^{-1}$. By Lemma \ref{lem:qadic}, getting $u\in \F_{q^k}^m$ from $v\in \F_q^n$ needs $\mathcal{O}_q(mk)=\mathcal{O}_q(n)$ operations. Analogously to the normalization in the proof of Theorem \ref{thm1}, the normalization of $u$ requires $\mathcal O_q(k^2 m)=\mathcal O_q(kn)$ operations. Then the complexity of the computation of  $\epsilon(i)$ become negligible, as does the computation of $\phi_{k,\epsilon(i)}(u')$, by Lemma~\ref{lem:compphi}. 
\end{proof}



\subsection{The Enumerative Coding Point of View}\label{sec:enuDes}

In this subsection we want to use the idea of enumerative coding to derive a message encoding and a retrieval map for Desarguesian spread codes.

The method of enumerative coding for block codes $C\subseteq \F_q^n$, as presented by Cover in \cite{co73}, is as follows: Denote by $\mathrm{enu}_C(u_1,\dots,u_j)$ the number of elements of $C$, for which the first $j$ coordinates are given by $(u_1,\dots,u_j)$. Moreover, fix an order $<$ on $\F_q$. Then the lexicographic index of $u=(u_1,\dots,u_n)\in \F_q^n$ is given by $\mathrm{ind}_C(u)=\sum_{j=1}^n \sum_{y<u_j} \mathrm{enu}_C(u_1,\dots,u_{j-1},y) $. This indexing function is then a message retrieval map. 

This method of enumerative coding was adapted to the whole Grassmannian space in \cite{si11}. In their setting, the enumerating function counts all vector spaces in $\G$, whose reduced row echelon forms fulfill certain conditions. 

For our purpose of deriving a message encoder for Desarguesian spread codes in $\G$, we use the idea of \cite{co73} on $\mathcal G_{q^k}(1,m)$. I.e.,  we define the enumerating function 
 $\mathrm{enu}_m(u_1,\dots,u_j)$ to count all elements of $\mathcal G_{q^k}(1,m)$, whose normalized representations have $(u_1,\dots, u_j)$ as their first $j$ entries. Furthermore we need to fix a bijection between $\F_{q^k}$ and the set of integers $\{0,\dots,q^k-1\}$, where we choose the map $\varphi$, as defined in Section \ref{sec:preliminaries}, as this bijection. This map induces an order $<$ on  $\F_{q^k}$.  
The indexing function for enumerative coding on $\mathcal G_{q^k}(1,m)$ is now given by
\begin{align*}
\mathrm{ind}:\quad  \mathcal G_{q^k}(1,m) &\longrightarrow \{0,\dots,(q^n-1)/(q^k-1) -1\}\\
\rs(u_1,\dots,u_m) &\longmapsto \sum_{j=1}^m \sum_{y<u_j} \mathrm{enu}_m(u_1,\dots,u_{j-1},y)   
\end{align*}
where $(u_1,\dots,u_m)$ is the normalized basis vector of the preimage. 
This function can easily be extended to a message retrieval map for Desarguesian spread codes, as shown in the following.

\begin{thm}\label{thm:enuenc}
The map $\mathrm{des}\circ \mathrm{ind}^{-1}$ is an encoding map for the Desarguesian spread code $\C \subseteq \G$ from Construction I. Its inverse  $ \mathrm{ind}\circ \mathrm{des}^{-1}$ is the corresponding message retrieval map.
\end{thm}
\begin{proof}
Since the Desarguesian spread code in $\G$ from Construction I is isomorphic to $\mathcal G_{q^k}(1,m)$, we simply need to show that $\mathrm{ind}$ is a bijection. This can be done in analogy to Cover's original proof in \cite{co73}: If $y_1 < y_2$, 
 then $\mathrm{ind}(u_1,\dots,u_{j-1},y_1, *,\dots,*) < \mathrm{ind}(u_1,\dots,u_{j-1},y_2, *,\dots,*)$ for any $y_{1},y_{2}\in \F_{q^{k}}$. Thus $\mathrm{ind}$ is injective. We now show that the image is $\{0,\dots,(q^n-1)/(q^k-1) -1\}$. The element in the preimage of lowest index is $\mathrm{ind}(0,0,\dots,0,1)$, which has index $0$. Let $y_{\max}$ be the largest element of $\F_{q^{k}}$ with respect to $<$. Then the preimage with the highest index is $(1,y_{\max},\dots,,y_{\max})$, for which $\mathrm{ind}$ counts all other elements of $\mathcal G_{q^k}(1,m)$. Hence $\mathrm{ind}$ takes on every value in $\{0,\dots,(q^n-1)/(q^k-1) -1\}$. I.e.,  $\mathrm{ind}$ is bijective, and since $\mathrm{des}: \mathcal G_{q^k}(1,m) \rightarrow \G$ is an injective map, whose image is the spread code, the statement follows.
\end{proof}

It remains to investigate how the map $\mathrm{enu}_m$ can be efficiently computed.

\begin{lem}\label{lem:enu}
Let $j\leq m$ and  $(u_1,\dots, u_j)\in \F_{q^k}^j$ such that the first non-zero entry, if existent, is equal to $1$. 
Then
$$\mathrm{enu}_m(u_1,\dots,u_j)= \left\{ \begin{array}{ll} 
\frac{q^{k(m-j)}-1}{q^k-1} & \textnormal{ if } u_1=\dots=u_{j}=0 \\ 
q^{k(m-j)}& \textnormal{ else }  \end{array} \right.  .$$
\end{lem}
\begin{proof}
If  $(u_1,\dots, u_j)$ is all-zero, then the remaining $m-j$ entries of any completion in $\mathcal G_{q^k} (1,m)$ need to be elements of $\mathcal G_{q^k}(1,m-j)$. The number of such elements is exactly $(q^{k(m-j)}-1)/(q^k-1)$. On the other hand, if $(u_1,\dots, u_j)$ is not all-zero, then any completion is already normalized (since we assume that  $(u_1,\dots, u_j)$ itself is normalized) and the remaining entries can be any element of $\F_{q^k}$. This implies the formula.
\end{proof}

We can now simplify the indexing function as follows.

\begin{cor}\label{cor10}
Let $\rs(u_1,\dots,u_m) \in \mathcal{G}_{q^k}(1,m)$. Define $j_1:=\min\{j\mid u_j\neq 0\}$ as the coordinate of the first non-zero entry of $(u_1,\dots,u_m)$. Then
\begin{align*} 
\mathrm{ind}(\rs(u_1,\dots,u_m)) 
&= \frac{q^{k(m-j_1)}-1}{q^k-1}  +  \sum_{j=j_1+1}^m \varphi(u_j) q^{k(m-j)} \\
&= \sum_{j=0}^{m-j_1-1} \left( \varphi(u_{m-j}) +1\right) q^{kj} .
\end{align*}
\end{cor}
\begin{proof}
Since we assume that all vectors are normalized, we have $u_{j_1}=1$ and $u_1=u_2= \dots = u_{j_1-1}=0$.
We can use Lemma \ref{lem:enu} and rewrite
\begin{align*}
 \sum_{j=1}^m \sum_{y<u_j} \mathrm{enu}_m(u_1,\dots,u_{j-1},y) &=  \sum_{j=1}^{j_1 } \sum_{y<u_j} \mathrm{enu}_m(u_1,\dots,u_{j-1},y) +  \sum_{j=j_1+1}^m \sum_{y<u_j} \mathrm{enu}_m(u_1,\dots,u_{j-1},y) \\
 &= \mathrm{enu}_m(\underbrace{0,0, \dots, 0,0}_{j_1}) +  \sum_{j=j_1+1}^m \sum_{y<u_j} \mathrm{enu}_m(u_1,\dots,u_{j-1},y)\\
 &= \frac{q^{k(m-j_1)}-1}{q^k-1}  +  \sum_{j=j_1+1}^m \varphi(u_j) q^{k(m-j)} \\
 &=  \sum_{j=0}^{m-j_1-1} q^{kj} +  \sum_{j=0}^{m-j_1-1}  \varphi(u_{m-j}) q^{kj} \\
&= \sum_{j=0}^{m-j_1-1} \left( \varphi(u_{m-j}) +1\right) q^{kj} .
\end{align*}
\end{proof}

\begin{ex}
Let $q=k=2$ and $m=3$. Moreover, let $\alpha\in \F_{2^2}$ be a primitive element and let the elements of $\F_{2^2}$ be identified via the map $\varphi$ with $0\mapsto 0, 1\mapsto 1, \alpha\mapsto 2, \alpha^2\mapsto 3$. 
Then
\begin{align*}
\mathrm{ind}( \rs(0, 0, 1)  ) = 0+0 =0    &\hspace{1cm} \mathrm{ind}( \rs(1, 0, 1)  ) = 5+1 =6\\ 
\mathrm{ind}( \rs(0, 1, \alpha)  ) = 1+2=3    &\hspace{1cm} \mathrm{ind}(\rs( 1, \alpha^2, 1)  ) = 5+12+1= 18   .
\end{align*}
Note that the two results on the right differ from those seen in Example \ref{ex6}.
\end{ex}

The following theorem and corollary show that the enumerative encoder is equal to the previously described encoder $\enc_1$, if we replace $\phi_{k,\epsilon(i)}$ with $\bar\phi_{k,\epsilon(i)}$ in the definition of $f$, where
\begin{align*}
\bar\phi_{k,\epsilon(i)} :  \F_{q^k}^{\epsilon(i)} & \longrightarrow  \{0,1, \dots, q^{k\epsilon(i) -1}\}  \\
(u_1, u_2, \dots, u_{\epsilon(i)}) & \longmapsto \sum_{j=0}^{\epsilon(i)-1} \varphi(u_{m-j})q^{kj} 
\end{align*}
is another bijection from $\F_{q^k}^{\epsilon(i)} $ to $\{0,1, \dots, q^{k\epsilon(i) -1}\}$. 
We call the new map, i.e., $f$ after replacing $\phi_{k,\epsilon(i)}$ with $\bar\phi_{k,\epsilon(i)}$, $\bar f$.

\begin{thm}
The map $\bar f^{-1}: \mathcal{G}_{q^k}(1,m) \rightarrow \{0,1,\dots, \frac{q^n-1}{q^k-1}\}$ is equal to the map $\mathrm{ind}: \mathcal{G}_{q^k}(1,m) \rightarrow \{0,1,\dots, \frac{q^n-1}{q^k-1}\}$. \end{thm}
\begin{proof}
Let $\rs(u_1,\dots, u_m)\in \mathcal{G}_{q^k}(1,m)$ and let $(u_1,\dots, u_m)\in\F_{q^k}^m$ be its normalized representation. Furthermore let $\epsilon(i)$ be defined as in the definition of $f$, and let $j_1$ be defined as in Corollary \ref{cor10}. Then $\epsilon(i) = m-j_1$ and 
$$ \frac{q^{k(m-j_1)}-1}{q^k-1} = \frac{q^{k\epsilon(i)}-1}{q^k-1} =  \sum_{j=0}^{\epsilon(i)-1} q^{kj} .$$ 
One knows that $\varphi(u_{m-\epsilon(i)})=\varphi(u_{j_1})=1$ and $\varphi(u_{m-j})=0$ for $j>\epsilon(i)$. 
Hence, together with Corollary \ref{cor10}, we get
$$\mathrm{ind}(\rs(u_1,\dots,u_m) )
= \sum_{j=0}^{m-j_1-1} \left(\varphi(u_{m-j})+1\right) q^{kj} =  \sum_{j=0}^{\epsilon(i)-1} q^{kj}  +  \sum_{j=0}^{\epsilon(i) -1} \varphi(u_{m-j}) q^{kj}  $$
$$=   \sum_{j=0}^{\epsilon(i)-1} q^{kj}  +  \bar \phi_{k,\epsilon(i)}(u_{m-\epsilon(i)+1},\dots,u_m) 
= \bar f^{-1}(\rs(u_1,\dots,u_m) ) .$$
\end{proof}

The previous theorem straightforwardly implies the following corollary, that the two encoders (and hence also the respective message retrieval maps) for Desarguesian spread codes from Subsections \ref{ssec:adhoc} and \ref{sec:enuDes} are equal, if we replace $f$ with $\bar f$ in the definition of $\enc_1$. We denote this second encoder by $\overline{ \enc_1} := \mathrm{des} \circ \bar f$.

\begin{cor}
Let $\C \subseteq \G$ be the Desarguesian spread code from Construction I. 
\begin{itemize}
\item
The encoding map $\overline{\mathrm{enc_1}} : \{0,1,\dots, \frac{q^n-1}{q^k-1}\} \rightarrow \G$ is equal to the encoding map $\mathrm{des}\circ \mathrm{ind}^{-1} : \{0,1,\dots, \frac{q^n-1}{q^k-1}\} \rightarrow \G$. 
\item 
The message retrieval map $\overline{\mathrm{enc_1}}^{-1} : \C   \rightarrow \{0,1,\dots, \frac{q^n-1}{q^k-1}\}$ is equal to the message retrieval map $\mathrm{ind}\circ \mathrm{des}^{-1} : \C   \rightarrow \{0,1,\dots, \frac{q^n-1}{q^k-1}\}$. 
\end{itemize}
\end{cor}


Even though the two encoders $\overline{\mathrm{enc_1}}$ and $\mathrm{des}\circ \mathrm{ind}^{-1}$ are equal as a map, they give rise to different ways of computing the encoding and the corresponding message retrieval map. Algorithm \ref{alg:Des2} describes an alternative way of computing the message retrieval for the Desarguesian spread code from Construction I, based on the idea of enumerative coding. The correctness of it follows from Theorem \ref{thm:enuenc} and Corollary \ref{cor10}. We will not give an alternative algorithm for the encoder, since computing the inverse of the indexing function, $\mathrm{ind}^{-1}$, does not give rise to an easier algorithm than Algorithm \ref{alg:DesEnc}.

\begin{algorithm}[ht]
\begin{algorithmic}
\REQUIRE{A spread codeword $\Uvs\in \C$.}
\STATE{Choose a non-zero vector $v\in \Uvs$.} 
\STATE{Compute  $u=(\psi_k(v_1,v_2,\dots,v_k),\psi_k(v_{k+1},v_{k+2},\dots,v_{2k}), \dots, \psi_k(v_{n-k+1},v_{n-k+2},\dots,v_n))$.}\STATE{Normalize $u$.}
\STATE{Set $j_1:= \min\{j\mid u_j\neq 0\}$.}
\STATE{Compute $i=\sum_{j=0}^{m-j_1-1} \left( \varphi(u_{m-j}) +1\right) q^{kj}$.}
\RETURN $i$
\end{algorithmic}
\caption{Message retrieval based on enumerative coding for the Desarguesian spread code  $\C \subseteq \G$ from Construction I.}
\label{alg:Des2}
\end{algorithm}

\begin{thm}
The computational complexity of Algorithm \ref{alg:Des2} is in $\mathcal{O}_q(kn)$.
\end{thm}
\begin{proof}
By Lemma \ref{lem:qadic}, getting $u\in\F_{q^k}^m$ from $v\in \F_{q}^n$ needs $\mathcal{O}_q(mk)=\mathcal{O}_q(n)$ operations. Analogously to the normalization in the proof of Theorem \ref{thm1}, the normalization of $u$ requires $\mathcal O_q(k^2 m)=\mathcal O_q(kn)$ operations. Since we represent $i$ in its $q$-adic expansion, we have $\varphi(u_i)=\psi^{-1}_k(u_i)$ for $i=1,2,\dots,m$. From Corollary \ref{cor10} we know that the ${q}$-adic expansion of $i$ is given by
\begin{align*}
& (\overbrace{\underbrace{10\dots 0}_k \underbrace{10 \dots 0}_ k \dots   \underbrace{10 \dots 0}_ k }^{(m-j_1)k}  \underbrace{00 \dots 0}_{kj_1})
+  (  \psi^{-1}_k(u_{m}) \, \psi^{-1}_k(u_{m-1}) \, \dots \, \psi^{-1}_k(u_{j_1+1}) \, \underbrace{00 \dots 0}_{kj_1} )  
\end{align*} 
Since we need at most $m$ computations of $ \psi^{-1}_k(u_j)$ and $n$ additions, this last step does not increase the overall complexity.
\end{proof}

Thus Algorithm \ref{alg:Des2} is an alternative to Algorithm \ref{alg:Des}, with the same complexity order.

\begin{rem}
Both, Algorithm \ref{alg:Des} and Algorithm \ref{alg:Des2}, can be improved if the input codeword is represented in row reduced echelon form. Then $v$ should not be a random element, but the first row of the input matrix. This implies that $u$ is already normalized, which improves the complexity of both message retrieval algorithms.
\end{rem}


\section{Message Encoding for Cyclic Orbit Codes}\label{sec:orbit}

Recall that an orbit code $\C\subseteq \G$ is defined as the orbit of a given $\Uvs\in \G$ under the action of a subgroup $G$ of $\GL_n$. In general it holds that $|\Cvs | \leq |G|$, and not necessarily $|\Cvs | = |G|$, i.e., some elements of $G$ might generate the same codewords. Denote by 
$$\stab_{\GL_n}(\Uvs) := \{A\in \GL_n \mid \Uvs A = \Uvs\}$$ 
the stabilizer of $\Uvs$ in $\GL_n$, and by $G/\stab_{\GL_n}(\Uvs)$ the set of all right cosets $\stab_{\GL_n}(\Uvs) A$ for $A \in \GL_n$.  
We define the map 
\begin{align*}
g  : \quad G/\stab_{\GL_n}(\Uvs)  &\longrightarrow \G \\
[A] &\longmapsto \Uvs A  .
\end{align*}
where $[A]$ denotes the coset of $A$.

\begin{thm}
The map $g$ is injective.
\end{thm}
\begin{proof}
Let $A,B \in G$. Assume that  $g(A)=g(B) \iff \Uvs A = \Uvs B$, then 
$$AB^{-1} \in \stab_{\GL_n}(\Uvs) ,$$
and thus $A = AB^{-1}B \in \stab_{\GL_n}(\Uvs) B$. Hence, $A$ and $B$ are in the same right cosets of $\stab_{\GL_n}(\Uvs) $, which implies the statement.
\end{proof}

For the remainder of this paper we will restrict ourselves to cyclic orbit codes, since these have simpler message encoders. Moreover, they have more useful structure than other orbit codes and are therefore  better understood from a construction and error decoding point of view.

Cyclic orbit codes are those codes that can be defined by the action of a cyclic subgroup $G$, i.e., $G=\langle P\rangle$ for some matrix $P\in \GL_n$. This notion is not to be mistaken with the definition of cyclic subspace codes from \cite{be14,et11}, which are unions of special cyclic orbit codes with different initial subspaces. For cyclic orbit codes one clearly has a bijection from $\cM =\{0,1,\dots, \ord(P)-1 \}$ to $G$, namely
\begin{align*}
 h' : \quad  \{0,1,\dots, \ord(P)-1 \} &\longrightarrow G  \\
i &\longmapsto P^i .
\end{align*}
From group theory (see e.g.\ \cite{ke99}) one knows that $|G/\stab_{\GL_n}(\Uvs)|$ is a divisor of $|G| = \ord(P)$ and that if $\ord_{\Uvs}(P):= |G/\stab_{\GL_n}(\Uvs)| < |G|$, then $\Uvs P^i = \Uvs P^{i+ \ord_{\Uvs}(P)}$. Thus it follows:

\begin{lem}
The map
\begin{align*}
 h : \quad  \{0,1,\dots, \ord_{\Uvs}(P)-1 \} &\longrightarrow G/\stab_{\GL_n}(\Uvs)  \\
i &\longmapsto [P^i] .
\end{align*}
is a bijection for any $\Uvs \in \G$.
\end{lem}

\begin{cor}\label{cor:enc4}
The map $\enc_2 := g \circ h$ is injective and hence is an encoding map from the message set $\cM =\{0,\dots, \ord_{\Uvs}(P)-1 \}$ to the cyclic orbit code $\Cvs = \Uvs \langle P \rangle \subseteq \Gr$.
\end{cor}

Note that $\enc_2$ can be computed straightforwardly, with matrix multiplications. Its inverse, i.e., the message retrieval map, is based on a discrete logarithm problem (DLP), which is in general known to be a hard problem. There are many results on when the DLP is hard and when it is not; for a survey of various algorithms and their complexities see e.g.\ \cite{od85}. In the following we will investigate some special types of cyclic orbit codes with respect to the computation of $\enc_2$ and $\enc_2^{-1}$.


\subsection{Primitive Cyclic Orbit Codes}

For this subsection let $\alpha$ be a primitive element of $\F_{q^n}$, $p(x) \in \F_q[x]$ its minimal polynomial and $P$ the corresponding companion matrix. Denote by $G = \langle P \rangle $ the group generated by it. Because of the primitivity it holds that
$$ \ord(\alpha) = \ord (P) = |G| = q^n-1 .$$
We call $\C = \Uvs G $ a \emph{primitive cyclic orbit code} for any $\Uvs\in \G$. For more information on the cardinality and minimum distance of different primitive cyclic orbit codes the interested reader is referred to \cite{gl14,tr11a}.
We can now state the message encoding algorithm for primitive cyclic orbit codes. For this let $U\in\F_q^{k\times n}$ be a matrix, such that $\rs(U) =\Uvs$.

\begin{algorithm}
\begin{algorithmic}
\REQUIRE{Message $i\in \{0,1,\dots, \ord_{\Uvs}(P) -1\}$.}
\STATE{Compute $P^i$.}
\STATE{Compute $V=UP^i$.}
\RETURN{$\Vvs =\rs(V)$}
\end{algorithmic}
\caption{Message encoding for a primitive cyclic orbit code $\mathcal C = \Uvs \langle P \rangle \subseteq \G$. }
\label{alg:cocEnc}
\end{algorithm}

\begin{thm}
The computational complexity of Algorithm \ref{alg:cocEnc} is in $\mathcal O_q(n^3)$.
\end{thm}
\begin{proof}
Since $P$ is a companion matrix of a primitive polynomial, we can compute $P^i$ as described in Lemma \ref{lemPexp}. Hence, this can be done with $ \mathcal O_q(n^3)$ operations. The multiplication with $U\in \F_q^{k\times n}$ can be done with $kn^2$ operations, thus the overall complexity is in $\mathcal O_q(n^3)$.
\end{proof}

\vspace{0.5cm}

For the message retrieval map we assume that an error correcting decoding algorithm has already found $P^i$ such that $\Uvs P^i$ is the respective codeword. This is a realistic assumption, as can be seen in the decoding algorithms of \cite{tr11a}. 
The retrieval map then needs to solve a discrete logarithm in the group $\langle P \rangle$ to find the exponent $i$, which corresponds to the message. The group $\langle P \rangle$ has order $q^n-1$. Since $\F_q[P]\cong\F_q[\alpha]$, we can equivalently compute the discrete logarithm in the group $\langle \alpha \rangle$.

There are several known algorithms to compute discrete logarithms. In this paper we will work with the well-known Pohlig-Hellmann algorithm, see e.g.\  \cite[Sec. 3.6.3]{me97b}.

\begin{lem}\label{lemPH}\cite{me97b}
The Pohlig-Hellman algorithm for computing a solution for the discrete logarithm in a group of order $q^n-1$ has a computational complexity in $\mathcal{O}_{q^n}(\sum_{i=1}^r e_i(\log_2 q^n+\sqrt{p_i}))$, where $\prod_{i=1}^r p_i^{e_i} $ is the prime factorization of $q^n-1$. 
\end{lem}

We describe a message retrieval algorithm for a primitive cyclic orbit code $\mathcal C = \Uvs \langle P \rangle \subseteq \G$ using the Pohlig-Hellman algorithm in Algorithm~\ref{alg:coc}. In the algorithm, $\rho: \F_q [\alpha] \rightarrow\F_q [P]$ denotes the isomorphism introduced in Section \ref{sec:preliminaries}.


\begin{algorithm}
\begin{algorithmic}
\REQUIRE{
A codeword $\Vvs \in \C$ and $ P^i$ such that $\Uvs P^i = \Vvs$. 
}
\STATE{Compute $\beta = \rho^{-1}(P^i)$.}
\STATE{Use the Pohlig-Hellman algorithm to find $i=\log_\alpha \beta$.}
\RETURN{$i$}
\end{algorithmic}
\caption{Message retrieval 
for a primitive cyclic orbit code $\mathcal C = \Uvs \langle P \rangle \subseteq \G$. }
\label{alg:coc}
\end{algorithm}

\begin{thm}\label{thm:Pohlig}
Let $\prod_{i=1}^r p_i^{e_i} $ be the prime factorization of $q^n-1$. 
The computational complexity of Algorithm~\ref{alg:coc} is in
\[\mathcal{O}_{q}( n^3\log_2 q  \sum_{i=1}^r e_i +  n^2 \sum_{i=1}^r  e_i\sqrt{p_i}     ) \]
\end{thm}
\begin{proof}
Since $P$ is the companion matrix of $\alpha$, $\beta$ is simply $\phi_{1,n}$ of the first row of $P^i$, hence, by Lemma \ref{lem:compphi}, the computation of $\beta$ can be done with at most $n$ operations over $\F_q$. 
Then it follows from Lemma \ref{lemPH} that the discrete logarithm can be computed with a complexity in $\mathcal{O}_{q}( n^3\log_2 q  \sum_{i=1}^r e_i +  n^2 \sum_{i=1}^r  e_i\sqrt{p_i}     ) $. 
Since any operation in $\F_{q^n}$ can be done with at most $\mathcal{O}_q(n^2)$ operations over $\F_q$, the statement follows.
\end{proof}

We can upper bound this complexity for cyclic orbit codes in $\G$, in the case where $q^n-1$ is $n^2$-smooth as follows.

\begin{cor}
If $q^n-1$ is $n^2$-smooth (i.e., if all prime factors of $q^n-1$ are less than or equal to $n^2$) and all $e_i$ are less than or equal to $k$, then the complexity order of Algorithm~\ref{alg:coc}
is upper bounded by $\mathcal{O}_q(n^3 k r \log_2 q )$, where $r$ is the number of distinct prime factors of $q^n-1$.
\end{cor} 

From an application point of view a complexity order of at most $\mathcal{O}_q(n^3 k r \log_2 q)$ is reasonable if $r$ is upper bounded by $n$. If we assume e.g.\ that $q\leq 2^k$, then we can simplify the above complexity order to $\mathcal{O}_q(n^3  k^2 r)$. For comparison, the complexities of the decoders in \cite{ko08,si08j} are at least cubic in $n$ and the decoding complexity of the rank-based error decoding algorithm for primitive cyclic orbit codes in \cite{tr11a} is of order $\mathcal{O}_q(q^{2k}(n^2 + k^2 n))$.

The following question remains: for which values of $q$ and $n$ is  $q^n-1$ $n^2$-smooth? For $q=2$ and $q=3$ 
Tables \ref{table1} and \ref{table2} show values of $n \leq 60$ for which $q^n-1$ is $n^2$-smooth. 
As explained before, the Pohlig-Hellman algorithm for these cases has a complexity of order at most $\mathcal O_q(n^3 k^2 r)$ if $k$ is at least $\max\{e_i\mid i=1,\dots,r\}$. As one can see, the largest respective exponents $e_i$ for the values presented in Table \ref{table1} are less than or equal to $3$, hence $k$ should be at least $3$. This is not much of a restriction, since for $k\leq 2$, any constant dimension code in $\Gr$ is no-error-correcting. In Table \ref{table2} the values for $e_i$ are larger, hence the restriction on $k$ is stricter. But since we are only considering complexity orders, we can allow $e_i$ to be slightly greater than $k$, without impairing the overall complexity order. Furthermore, we can see that $r$ is reasonably small for these parameter sets, which is is necessary for an efficient performance of the Pohlig-Hellman algorithm.

\begin{table}[ht]
\begin{center}
\begin{tabular}{|c|c|c|c|c|c|c|}
\hline
$n$ & 
$\max p_i$ & $\max e_i$ &  $ \max(e_i n, e_i p_i)$ & $r$ & $n^2$ \\
\hline
$6$ 
& $7  $ & 2 & 18 & 2& 36\\
$8$ 
& 17 & 1 & 17 & 3 & 64\\
$9$ 
& 73 &1& 73 & 2 & 81\\
$10$ 
& 31 &1& 31 & 3 & 100\\
$11$ 
& 89 &1 & 89 & 2 & 121\\
$12$ & 13 & 2& 24 & 4 &144 
\\
14 & 127& 1& 127 & 3 & 196
\\
15& 151 & 1 & 151  & 3 & 225
\\
18& 73 & 3 & 73 & 4 & 324
\\
20& 41 & 2 & 41 & 5 & 400
\\
21& 337 & 2 & 337& 3 & 441 
\\
24& 241 & 2 &  241 & 6 & 576
\\
28& 127 & 1 & 127& 6 &  784
\\
30& 331 & 2 & 331 & 6 & 900
\\
36& 109 & 3 & 109 & 8 & 1296
\\
48 & 673 & 2 & 673 & 9 & 2304
\\
60 & 1321 & 2 & 1321 & 11 & 3600
\\
\hline
\end{tabular}
\caption{Values of $n\leq 60$ for which $2^n-1 = \prod_{i=1}^r p_i^{e_i}$ is $n^2$-smooth.}\label{table1}
\end{center}
\end{table}

\begin{table}[ht]
\begin{center}
\begin{tabular}{|c|c|c|c|c|c|}
\hline
$n$ & 
$\max p_i$ & $\max e_i$ &  $ \max(e_i n, e_i p_i)$ &$ r$ & $n^2$ \\
\hline
6 &13 & 3 & 18 & 3 & 36\\
8&41&5&41& 3 & 64\\
10&61&3&61& 3 & 100\\
12&73&4&73& 5 &144\\
16&193&6&193& 5 & 256\\
\hline
\end{tabular}
\caption{Values of $n\leq 60$ for which $3^n-1 = \prod_{i=1}^r p_i^{e_i}$ $n^2$-smooth.}\label{table2}
\end{center}
\end{table}

\begin{rem}
There are values for $q$ and $n$, where $q^n-1$ is not $n^2$-smooth but the largest prime factor of $q^n-1$ is close to $n^2$. Then the complexity order of Algorithm~\ref{alg:coc} will still be $\mathcal O_q(n^3 k^2 r)$ and the message retrieval algorithm will thus still not increase the overall decoding complexity for many parameters. 
\end{rem}


\subsection{Unions of Primitive Cyclic Orbit Codes}

First we want to generalize the previously discussed decoding algorithm to unions of primitive cyclic orbit codes. 
As seen in \cite{et11,gl14,ko08p}, unions of primitive cyclic orbit codes are among the best known constructions for constant dimension codes. Decoding such codes can be done by using a decoding algorithm for single orbits for each of the orbits that constitute the code. In the error correction decoding process the algorithm needs to decide which orbit the respective closest codeword is on. This information can then be passed on to the message retrieval algorithm, which simply applies Algorithm~\ref{alg:coc} on that chosen orbit.

As before let $\alpha\in \F_{q^n}$ be a primitive element and $P$ its companion matrix. The various orbits that form our code as a union are all generated by the action of $P$ (respectively $\alpha$). They are given by the initial points $\Uvs_1,\dots,\Uvs_z\in \G$. For simplicity we assume that all orbits have the same cardinality $c^*$. Then the following map is an encoding map for the code $\C = \bigcup_{i=1}^z \Uvs_i \langle P \rangle \subseteq \G$: 
\begin{align*} 
\enc_3 : \{0,\dots,  z  c^*-1\} &\longrightarrow \G \\
i &\longmapsto  \Uvs_j P^i \quad, \quad j=\left\lceil \frac{i+1}{c^*}\right\rceil
\end{align*}
Note that, because of the cardinality of the orbits, $\Uvs_j P^i= \Uvs_j P^{i+c^*}$. Therefore we can also compute $P^\ell$, where $\ell\equiv i \mod c^*$, instead of $P^i$ in the above encoding map. For the inverse map, i.e., the message retrieval, we require the information $j$ (which orbit the word is on) from the error decoder. Then we use Algorithm~\ref{alg:coc} to retrieve the respective exponent $\ell$ of $P$. With both these pieces of information we can easily reconstruct the message $i=\ell + (j-1)c^*$.

\begin{ex}
Let $q=2, n=4, k=2$ and $\alpha\in \F_{2^4}$ be a root of the irreducible polynomial $x^4+x+1$ and denote by $P$ its companion matrix. Moreover, $\rho: \F_2[\alpha] \rightarrow \F_2[P]$ is the isomorphism, as previously described. 
One has $\ord(\alpha)=15$, i.e., $\alpha$ is primitive. Let the extension field representations of the initial points $\Uvs_1, \Uvs_2\in \mathcal G_2(2,4)$ be $\{0,1,\alpha,\alpha+1\}$ and $\{0,1,\alpha^2,\alpha^2+1\}$, respectively. Both orbits, $\Uvs_1\langle P\rangle$ and $\Uvs_2\langle P\rangle$, have cardinality $c^*=15$. Our code consists of the union of these two orbits, i.e., our message set is $\mathcal M = \{0,\dots,29\}$. Assume we received some $\Rvs \in \mathcal G_2(2,4)$ that is error decoded to the codeword $\Uvs_2 \rho(\beta)$ with $\beta= \alpha^3 + \alpha + 1$. We use Algorithm \ref{alg:coc} to compute $\ell = \log_\alpha \beta = 7 $ and retrieve the message $i= \ell + (j-1)c^* = 7 + 15 =22$.
\end{ex}


\subsection{Non-Primitive Irreducible Cyclic Orbit Codes}

The second generalization that we want to consider is the message retrieval of cyclic orbit codes $\Uvs \langle P \rangle \subseteq \G$  that are not generated by a companion matrix of a primitive polynomial, but rather of a non-primitive irreducible one. For more information on the distinctions of primitive, irreducible non-primitive and completely reducible cyclic orbit codes the interested reader is referred to \cite{gl14,tr11a}. 
For the irreducible (but not primitive) case we assume that $\alpha\in \F_{q^n}$ is irreducible of order less than (and a divisor of) $q^n-1$. Then $\F_{q^n}^*$ is partitioned into different orbits under the action of $\alpha$. The various vectors of a codeword possibly lie on several of these orbits, which the error correction algorithm needs to take into consideration. The message retrieval algorithm then only needs to compute the discrete logarithm in $\langle \alpha \rangle$, just as in the primitive case. The computation of the discrete logarithm is possibly easier than in the primitive case, since the group order is smaller.

\begin{ex}
Let $q=2, n=4, k=2$ and 
 $\alpha\in \F_{2^4}$ be a root of the irreducible polynomial $x^4+x^3+x^2+x+1$. As before, $\rho: \F_2[\alpha] \rightarrow \F_2[P]$ denotes an isomorphism. 
We have $\ord(\alpha)=5$, thus our message set is $\mathcal M = \{0,\dots,4\}$. $\F_{2^4}^*$ is partitioned into the orbits $\langle \alpha \rangle, (\alpha+1) \langle \alpha \rangle$ and $(\alpha^2 +1)\langle \alpha \rangle$. Let $\Uvs\in \mathcal G_2(2,4)$, $P$ be the companion matrix of $\alpha$ and $\C = \Uvs \langle P\rangle$ be the cyclic orbit we want to consider. 
Assume that we received some word $\Rvs \in \mathcal G_2(2,4)$, which was error decoded to the codeword 
$ \Uvs \rho(\beta)$ with $ \beta = \alpha^3 + \alpha^2 + \alpha + 1$. 
We use Algorithm~\ref{alg:coc} to get the message $i=\log_\alpha \beta = 4$.
\end{ex}

In the next example we illustrate two cases where the order of a primitive element is not $n^2$-smooth, but the order of some non-primitive irreducible element is $n^2$-smooth. This shows that there are more irreducible cyclic orbit codes with an efficiently computable  message retrieval map than only the primitive ones.
\begin{ex}
\begin{enumerate}
\item
Let $q=3,n=18,k=9$. Then any primitive element of $\F_{3^{18}}$ has order $3^{18}-1=2^3\cdot 7\cdot 13\cdot 19\cdot 37\cdot 757$, which is not $n^2$-smooth. But there also exists an irreducible element in $\F_{3^{18}}$ of order $(3^{18}-1)/(3^{9}-1) = 2^2\cdot 7\cdot 19\cdot 37$, which is $n^2$-smooth.
\item
Let $q=3,n=48,k=24$. Then any primitive element of $\F_{3^{48}}$ has order $3^{48}-1$, whose largest prime power is $6481$. Hence this order is not $n^2$-smooth. But there also exists an irreducible element in $\F_{3^{48}}$ of order $(3^{48}-1)/(3^{24}-1) = 2\cdot 17\cdot 97\cdot 193\cdot 577\cdot 769$, which is $n^2$-smooth.
\end{enumerate}
\end{ex}

\subsection{Completely Reducible Cyclic Orbit Codes}

Lastly we briefly explain how the previous results can be generalized to completely reducible cyclic orbit codes. For these codes the generating matrix $P$ is not a companion matrix of some irreducible polynomial, but rather a completely reducible matrix. I.e.,   $P$ can be brought into block diagonal form, where each block is again a companion matrix of some irreducible polynomial. For simplicity we assume that $P\in \F_{q}^{n\times n}$ is of the form
$$ P = \left( \begin{array}{ccccccccc}  P_1 &&&& \\ & P_2&&& \\ &&&\ddots& \\ &&&&P_t  \end{array}\right)    ,$$
where $P_i\in \F_q^{n_i\times n_i}$ is a companion matrix of a primitive element in $\F_{q^{n_i}}$. One can easily see that $n=n_1+n_2+\dots +n_t$. 
Then any element $u\in \F_q^n$ can be represented as an element of $\F_{q^{n_1}}\times \F_{q^{n_2}} \times \dots \times \F_{q^{n_t}}$. Furthermore, we can represent any element of the code this way, thus each of these blocks of length $n_i$ can be seen as a primitive cyclic orbit code. 
For a more detailed explanation of completely reducible orbit codes see \cite{tr11a}. The message encoding and retrieval can now be done in each of the blocks of length $n_i$. For the retrieval one gets $t$ integer solutions $i_1,\dots, i_t$. The final solution $i$, such that $\Uvs P^i$ is the received word, is then given by solving the system of simultaneous congruences $i\equiv i_j \mod n_j$ for $j=1,\dots,t$.

\begin{rem}
The linkage construction for cyclic orbit codes from \cite{gl14} is closely related to unions of completely reducible orbit codes, see \cite[Proposition 5.3]{gl14}. The message retrieval map explained before can be extended to work for the linkage construction as well.
\end{rem}


\section{A Hybrid Encoder for Semi-Linearly Isometric Codes}\label{sec:hybrid}

We have seen in the previous section that there exist parameters for which $\enc_2$ is a message encoding function for orbit codes, that has an efficient inverse map, i.e., an efficient corresponding retrieval map. For many parameters though, the procedures described in Section \ref{sec:orbit} are not efficiently computable. In this section we show how semi-linear isometry can be useful for message encoding and retrieval purposes. We will describe the results in general, for any pair of semi-linearly isometric codes, and then explain how this can be applied to the special class of spreads constructed as primitive cyclic orbit codes, since these are always (semi-linearly isometric to) Desarguesian spreads (cf.\ e.g.\ \cite[Theorem 14]{ba11a}).

Let $A\in \GL_n$ and $\sigma \in \mathrm{Aut}(\F_{q})$ a field automorphism. If $\C_1,\C_2 \subseteq \G$ are two constant dimension codes, such that $\sigma(\C_1 A)  = \C_2$ (as sets of vector spaces), where $\sigma$ is applied element-wise on the codewords of $\C_1 A$, then we call $\C_1$ and $C_2$ \emph{semi-linearly isometric}. If $\sigma = \mathrm{id}$, then we call the codes \emph{linearly isometric}. This name \emph{isometric} arises because the codes have the same cardinality and distance distribution. For more information on semi-linear isometry of subspace codes the interested reader is referred to \cite{tr13phd,tr12}. 

Assume that there exists an encoder $\mathrm{enc}$ 
for the code $\C_1$ (for a message set $\mathcal M$). Then
\begin{align*}
 \enc' : \quad \cM &\longrightarrow  \G \\
 i & \longmapsto \sigma(\enc(i) A) 
 \end{align*}
is an encoder for $\C_2$. We call this a $\emph{hybrid encoder}$ for $\C_1$ and $\C_2$. 

\begin{thm}
 Let $\enc$ and $\enc'$ be as above. Denote the complexity order of $\enc$ by $\omega(\enc)$. Then the complexity of the hybrid encoder $\enc'$ is in $\mathcal O_q(\omega(\enc)+kn^2).$ 
\end{thm}
\begin{proof}
Follows straightforwardly from the fact that the computational complexity order of the multiplication with $A$ is in $\mathcal O_q(kn^2)$. The complexity of the field automorphism is negligible.
\end{proof}

As an example we want to show how the idea of a hybrid encoder can be used for Desarguesian spread codes. As mentioned above, any spread constructed as a primitive cyclic orbit code is a Desarguesian spread, in the more general definition of Desarguesian spread. It follows that such a primitive cyclic orbit code is semi-linearly isometric to a code from Construction~I (see also  \cite[Corollary 16]{ba11a}).
With this knowledge we can use $\enc_1$ to efficiently encode and retrieve messages, but use the orbit code structure for an error correcting decoding algorithm, e.g. the coset leader decoding algorithm from \cite{tr11a}.

 \begin{ex}
 Let $\C_1$ be the spread constructed in Example \ref{ex1} and let $\C_2$ be the orbit spread code constructed in Example \ref{ex2}, both subsets of $\mathcal G_2(2,4)$ with five elements. 
 Then we can use the algorithm of Feulner from \cite{thomas}\footnote{This algorithm requires two codes in $\G$ as input and then computes if they are linearly isometric; and if so, finds the linear transformation from one code into the other.} to find a linear transformation $A\in \GL_4$, such that $\C_1 A= \C_2$. One such linear transformation is given by
 $$ A=\left(\begin{array}{cccc} 1&0&0&0 \\ 0&1&1&0 \\ 1&1&0&0 \\ 0&1&0&1 \end{array}\right) .$$
 Let $\beta$ be a primitive element of $\F_{2^4}$. 
 In the isomorphic extension field representation, $A$ maps the basis $\{1,\beta, \beta^2, \beta^3\}$ of $\F_{2^4}$ over $\F_2$ to the new basis $\{1, \beta^2+ \beta, \beta+1, \beta^3+ \beta\}$. We can now use $\C_1$ for message encoding, say we encode a given message to the codeword $\Uvs=\psi^{-1}_4\{0,1,\beta,\beta+1 \} \in \C_1$, then we send $\Uvs A= \psi^{-1}_4\{0,1,\beta^2+\beta,\beta^2+\beta+1 \}\in \C_2$ over the channel. We can then do error correction decoding in the code $\C_2$ with e.g.\ the coset leader decoder. Say we decoded the received word to the sent codeword $\Uvs A \in \C_2$. Then it suffices to apply $A^{-1}$ on only one of the non-zero elements of $\Uvs A$, e.g.\ $ \psi_4^{-1}(\beta^2+\beta)A^{-1} =  \psi_4^{-1}(\beta) = ( 0 , 1 , 0 , 0 )$ to identify the corresponding codeword $\Uvs \in \C_1$, from which we can then easily get the message as explained in Section \ref{sec:spread}.
 \end{ex}

\section{Conclusion}\label{conclusion}

In this work we investigate how message encoding can be done for spread and cyclic orbit codes, two families of subspace codes that have been well-studied for error correction in random network coding. 

We show that for Desarguesian spread codes in $\G$ there exists an encoding map such that the map itself and the inverse map are efficiently computable with a computational complexity of order at most $\mathcal O_q(kn)$. In addition, we study the method of enumerative coding for this family of codes and show that the first message retrieval map, with a little twist, is equal to the indexing function of enumerative coding.

Furthermore, we develop an encoder for general cyclic orbit codes. This map is efficiently computable, but the inverse, i.e., the message retrieval map, involves the computation of a discrete logarithm. This is known to be computationally hard in general, but we show for which parameters the Pohlig-Hellman algorithm computes the discrete logarithm in complexity of order at most $\mathcal O_q(n^3 k r \log_{2}q)$. This is done in detail for primitive cyclic orbit codes. Moreover, some remarks on how to generalize these results to unions of cyclic orbit codes and completely reducible cyclic orbit codes are given.

In the end we propose a hybrid encoder for semi-linearly isometric codes, which is useful if one knows an efficient message retrieval map for a linearly isometric code to a given one. We show how this can be realized for cyclic orbit codes that are  Desarguesian spreads, such that one can use the orbit structure for error correction, but avoid the discrete logarithm problem in the message retrieval part.

An open question for further research is, if there are other, for certain parameter sets more efficient, ways to solve the discrete logarithm problem in the message retrieval of orbit codes. Moreover, it would be interesting to find 
other families of semi-linearly isometric codes where a hybrid encoder can be helpful to combine efficient error correction decoders with efficient message retrieval maps.

\section*{Acknowledgment}
The author would like to thank Yuval Cassuto for his reference to enumerative coding, John Sheekey for his advice on Desarguesian spreads, and  Margreta Kuijper for fruitful discussions and comments on this work. She would furthermore like to thank the anonymous reviewers for their valuable comments.

\bibliographystyle{plain}

\bibliography{network_coding_stuff}

\end{document}